\documentclass[11pt,a4paper]{article}
\usepackage[utf8]{inputenc}
\usepackage[T1]{fontenc}
\usepackage[english]{babel}
\usepackage{lmodern}
\usepackage{microtype}
\usepackage[a4paper,top=2.5cm,bottom=2.5cm,left=2.7cm,right=2.7cm,footskip=1.2cm]{geometry}
\usepackage{amsmath,amssymb,amsthm,mathtools,bm}
\usepackage{graphicx,booktabs,array,tabularx,multirow,makecell,longtable}
\usepackage{caption,subcaption}
\usepackage{xcolor,enumitem,csquotes}
\usepackage{tikz}
\usetikzlibrary{arrows.meta,positioning,fit,calc,shapes.geometric,decorations.pathreplacing}
\usepackage{siunitx}
\newcolumntype{L}[1]{>{\raggedright\arraybackslash}p{#1}}
\newcolumntype{Y}{>{\raggedright\arraybackslash}X}
\usepackage[backend=biber,style=authoryear-comp,natbib=true,url=true,doi=true,isbn=false,maxbibnames=99,maxcitenames=2,uniquelist=false,sorting=nyt]{biblatex}
\usepackage[colorlinks=true,linkcolor=black,citecolor=black,urlcolor=black]{hyperref}
\usepackage[noabbrev,capitalize]{cleveref}
\theoremstyle{plain}
\newtheorem{proposition}{Proposition}

\theoremstyle{definition}
\newtheorem{definition}{Definition}

\newtheorem{criterion}{Criterion}
\theoremstyle{remark}

\crefname{assumption}{Assumption}{Assumptions}
\Crefname{assumption}{Assumption}{Assumptions}
\crefname{definition}{Definition}{Definitions}
\Crefname{definition}{Definition}{Definitions}
\crefname{criterion}{Criterion}{Criteria}
\Crefname{criterion}{Criterion}{Criteria}
\setlist{nosep,leftmargin=*}
\definecolor{seriesblue}{RGB}{31,78,121}
\definecolor{seriesgray}{RGB}{242,244,247}
\tikzset{seriesbox/.style={draw=black!70,rounded corners=2pt,align=center,inner sep=5pt,minimum height=9mm,fill=seriesgray},seriesarrow/.style={-{Latex[length=2.2mm]},thick,draw=black!75}}

\hypersetup{
  pdftitle={A Taxonomy of Event-Linked Perpetual Futures: Design Axes, Failure Modes, and Empirical Evaluability},
  pdfauthor={Maksym Nechepurenko},
  pdfsubject={Event-Linked Perpetual Futures; Taxonomy; Multi-Leg Derivatives; Finality; Empirical Evaluability},
  pdfkeywords={prediction markets, perpetual futures, event contracts, conditional probability, event spreads, baskets, entropy, liquidity, rolling contracts}
}

\title{A Taxonomy of Event-Linked Perpetual Futures:\\
Design Axes, Failure Modes, and Empirical Evaluability}
\author{Maksym Nechepurenko\thanks{Founder and Director of Research, ForesightFlow, the Research Department of Devnull FZCO, Dubai, United Arab Emirates. Email: \texttt{maksym@devnull.ae}.}}
\date{July 19, 2026}

\begin{document}
\maketitle
\begin{abstract}
The label \emph{event-linked perpetual} often conflates mathematically different contracts. A derivative on one binary probability is not the same object as a conditional ratio, spread, basket, path functional, liquidity index, rolling sequence, or flow-only swap. Treating them as variants of one risk engine obscures their failure modes.

This paper replaces a flat product list with a four-axis taxonomy: underlying geometry, temporal structure, settlement structure, and venue-oracle composition. It derives a corrected inheritance map from the single-event case and proves six narrow results. Terminal shortfall depends on exact settlement support and collateral, not on a generic ``bounded-event'' label. Conditional ratios become locally ill-conditioned as the conditioning value approaches zero. After one leg of a spread becomes final, the remaining exposure is an affine single-leg position. Basket risk depends on feasible joint outcomes and weights, not leg count. Entropy settles deterministically to zero at binary finality, whereas realized variation requires an explicit sampling and terminal-jump convention. Fixed-path replay does not identify deployment behavior when a liquidity contract changes its own reference process.

No new empirical estimates are reported. The paper assigns evidence grades and minimum data requirements to each design. The central conclusion is that no universal event-perpetual engine exists: every contract must separately define support, clocks, settlement, and source hierarchy, and its controls must satisfy all necessary conditions generated by those choices.
\end{abstract}


\clearpage
\section{Introduction}
\label{sec:introduction}

A ``perpetual on an election'' can describe at least four different balance sheets. It may track one candidate's binary claim, the spread between two candidates, a weighted basket of state-level markets, or a rolling sequence of election-related contracts. A fifth proposal may pay only a periodic flow derived from disagreement and have no terminal principal settlement. These contracts can share an interface while having different mathematical objects underneath it.

The distinction matters because a risk control is inherited only when the property that justified it is inherited. A terminal-support test is relevant to a contract with a terminal cash value. It is not automatically relevant to a flow-only swap. A boundary-aware funding rule for an index in $[0,1]$ does not transfer unchanged to a spread in $[-1,1]$. A halt tied to one event time does not specify what happens when two legs finalize asynchronously. A fixed-path replay is informative when the underlying path can plausibly be treated as invariant to the overlay; it is not informative for a derivative whose own trading changes the liquidity index being measured.

\subsection{Research question}

The paper asks:

\begin{quote}
Which structural properties define an event-linked derivative, which failure modes follow from those properties, and what evidence is required before a proposed design can be evaluated empirically?
\end{quote}

The answer is organized around axes rather than product names. A contract is represented by its underlying geometry, temporal rule, settlement rule, and venue-oracle composition. Canonical product examples are then points in this design space, not mutually exclusive categories.

\subsection{Relationship to the single-event case}

The revised first paper in this series treats the single binary probability-index overlay as a falsification problem rather than a validated product design \citep{paper1_revision_2026}. Two corrections are load-bearing here. First, terminal solvency is determined by the full entry-to-terminal account identity and the attainable settlement support. Second, uniform relative-basis pressure near zero or one is incompatible with bounded, economically payable funding transfers. This paper inherits those results only where their assumptions survive.

The original taxonomy overstated inheritance in several places. It treated all multi-leg bounded contracts as direct instances of the single-event collateral proposition, grouped entropy and variance despite different terminal behavior, and placed rolling and funding-only structures on the same axis as underlying geometry. The present revision removes those claims and proves narrower results.

\subsection{Contributions}

The revision makes six contributions.

\begin{enumerate}[label=(\roman*)]
  \item It defines a four-axis specification tuple that separates geometry, time, settlement, and source composition.
  \item It gives a support-dependent terminal-shortfall test that applies across variant designs without assuming one leverage convention.
  \item It proves variant-specific results for conditional sensitivity, asynchronous spreads, basket support, entropy collapse, and reflexive replay.
  \item It distinguishes true underlying families from modifiers: rolling is temporal; funding-only is a flow settlement convention; cross-venue construction is a source-composition choice.
  \item It supplies a corrected control-inheritance matrix showing where Paper~1 controls apply, require redesign, or are inapplicable.
  \item It specifies an empirical-evaluability ladder and minimum data package for each canonical design without presenting new empirical estimates.
\end{enumerate}

\subsection{Evidence and scope}

This is a formal taxonomy and research-design paper. It uses previously documented properties of Polymarket-class central limit order book (CLOB) venues only to motivate observability constraints. It does not rerun the Polymarket archive, estimate demand for any variant, or claim that a listed design would sustain an equilibrium. Paper~3 studies manipulation and regulation; Paper~4 studies fill-side participant behavior and quote-attribution limits \citep{paper3_manipulation_2026,paper4_microstructure_2026}.

\subsection{Organization}

\Cref{sec:base-contract} states the base account identity. \Cref{sec:design-axes} defines the taxonomy. \Cref{sec:conditional,sec:spread,sec:basket,sec:path-functionals,sec:liquidity,sec:rolling-flow} analyze canonical designs. \Cref{sec:crosscutting} gives the corrected inheritance map, and \cref{sec:evaluability} specifies empirical requirements.

\section{Related Design Literature}
\label{sec:related-work}

The taxonomy intersects several established literatures, but no one literature supplies the complete contract specification.

Combinatorial information markets provide the natural foundation for conditional and logically linked claims \citep{hanson_2003}. They solve a spot-market expressiveness problem; a leveraged overlay adds collateral, basis, and finality questions. Crypto perpetual futures provide the funding and non-expiring trading convention \citep{he_manela_schwert_2024}; their usual continuously tradable reference asset is precisely the assumption under examination in event-linked designs.

Basket derivatives provide tools for weighted aggregation and dependence-sensitive risk \citep{krekel_etal_2004}. Variance-swap and volatility-trading research clarifies that a path functional requires an explicit sampling and replication convention \citep{carr_madan_1998,demeterfi_etal_1999}. Information theory supplies the entropy transform and its endpoint properties \citep{cover_thomas_2006}. Market-liquidity and empirical-microstructure research explains why a liquidity statistic is endogenous to funding conditions, order placement, and execution \citep{brunnermeier_pedersen_2009,hasbrouck_2007}.

The present contribution is not a new pricing formula for each family. It is a contract-specification and evaluability framework that identifies where these adjacent literatures can be inherited and where event finality creates an additional state variable.

\section{Base Contract and Corrected Risk Inheritance}
\label{sec:base-contract}

Prediction-market prices need not equal objective probabilities because wealth, risk preferences, fees, and market structure enter the traded price \citep{wolfers_zitzewitz_2004,manski_2006}. Throughout, an event-market price is an observable reference value. The taxonomy does not require a probabilistic interpretation unless a contract definition explicitly uses one.

\subsection{Terminal settlement support}

Let $Z$ be the terminal settlement value of a derivative and let its attainable support be
\[
  \mathcal Z=\{z:\text{$z$ is permitted by the contract's outcome and settlement rules}\}.
\]
Write $z_{\min}=\inf\mathcal Z$ and $z_{\max}=\sup\mathcal Z$, and assume the entry mark $q_0\in[z_{\min},z_{\max}]$. Let $x>0$ be position size. Let $C_\tau^{\mathrm{net}}$ be initial collateral plus external top-ups, less withdrawals, funding, fees, and other cash transfers, but excluding the position's own price profit and loss. Intermediate variation-margin transfers may change the timing of cash, but their sum with the terminal mark change equals the entry-to-terminal price change.

\begin{proposition}[Support-dependent terminal shortfall]
\label{prop:support-shortfall}
For a long position, terminal equity is
\[
  E_\tau^{\mathrm{long}}=C_\tau^{\mathrm{net}}+x(Z-q_0).
\]
For a short position,
\[
  E_\tau^{\mathrm{short}}=C_\tau^{\mathrm{net}}+x(q_0-Z).
\]
If $z_{\min}$ and $z_{\max}$ are attainable, the worst-case long and short price losses from entry are $x(q_0-z_{\min})$ and $x(z_{\max}-q_0)$, respectively. A terminal account shortfall occurs exactly when net collateral is below the applicable adverse price loss.
\end{proposition}

\begin{proof}
For any pre-terminal mark $q_{\tau^-}$, long price profit and loss decomposes as
\[
  x(q_{\tau^-}-q_0)+x(Z-q_{\tau^-})=x(Z-q_0),
\]
and the short identity has the opposite sign. Minimizing the long payoff over $\mathcal Z$ selects $z_{\min}$; minimizing the short payoff selects $z_{\max}$. Negative equity is therefore equivalent to net collateral being smaller than the applicable adverse entry-to-terminal price loss.
\end{proof}

The proposition is deliberately more general than a leverage slogan. A probability-index long with $z_{\min}=0$ and collateral $xq_0/L$ has adverse-outcome shortfall for every $L>1$ if it survives without top-up, close, conversion, or guarantee. A spread or basket requires its own support and collateral convention. The number of legs does not determine solvency.

\subsection{Boundary funding}

For a normalized underlying $I_t\in[0,1]$, Paper~1 shows that a transfer bounded by payer equity cannot enforce a penalty uniformly proportional to relative basis $|q_t-I_t|/\min\{I_t,1-I_t\}$ arbitrarily close to zero or one \citep{paper1_revision_2026}. The result is inherited only when:

\begin{enumerate}[label=(\alph*)]
  \item the contract has an economically meaningful normalized boundary;
  \item the design objective is relative-basis control near that boundary; and
  \item transfers are required to remain finite and payable.
\end{enumerate}

A spread crossing zero does not inherit this result merely because a denominator could be chosen to vanish at zero. That would be a design-created singularity, not a natural boundary of the support.

\subsection{Five distinct risk channels}

A variant can fail through several channels that should not be collapsed into ``resolution risk.''

\begin{table}[ht]
\centering
\caption{Risk channels used throughout the taxonomy. The channels can coexist and require different controls.}
\label{tab:risk-channels}
\small
\begin{tabularx}{\textwidth}{L{2.8cm}YY}
\toprule
Channel & Question & Typical control object \\
\midrule
Support and collateral & Can the adverse terminal value exceed account collateral? & Exact support, margin, conversion, guarantee \\
Execution & Can exposure be reduced in available depth before a state transition? & Depth-aware limits, auctions, reduce-only \\
Definition & Is the payoff defined in every admissible state? & Void, fallback, roll, or conversion rule \\
Finality & Which leg, oracle, or protocol state is binding and when? & State machine and per-leg clocks \\
Reflexivity & Does trading in the derivative alter its own reference process? & Endogeneity model, position limits, alternative evidence \\
\bottomrule
\end{tabularx}
\end{table}

A complete variant specification must identify every active channel before selecting a margin or funding rule.

\section{A Four-Axis Taxonomy}
\label{sec:design-axes}

\begin{definition}[Event-linked derivative specification]
\label{def:specification}
An event-linked derivative specification is a tuple
\[
  \mathfrak C=(g,\mathcal T,\mathcal S,\mathcal C),
\]
where:
\begin{enumerate}[label=(\roman*)]
  \item $g$ is the underlying geometry mapping observable leg states and, where relevant, their history into a contract index;
  \item $\mathcal T$ is the temporal rule defining active legs, event times, rolls, and observation windows;
  \item $\mathcal S$ is the settlement rule defining cash flows, terminal values, conversions, voids, and fallback states; and
  \item $\mathcal C$ is the composition rule identifying venues, adapters, oracles, source hierarchy, and finality dependencies.
\end{enumerate}
\end{definition}

A user-interface label is not a specification. Two contracts with the same $g$ but different $\mathcal S$ can have different terminal support. Two contracts with the same $g$ and $\mathcal S$ but different $\mathcal C$ can have different latency and manipulation surfaces.

\subsection{Axis 1: underlying geometry}

Canonical geometries include:

\begin{itemize}
  \item single probability $p_t^{(A)}$;
  \item conditional ratio $u_t/v_t$;
  \item spread $p_t^{(A)}-p_t^{(B)}$;
  \item weighted basket $\sum_i w_i p_t^{(i)}$;
  \item path functional, such as realized variation or entropy;
  \item liquidity or another microstructure index.
\end{itemize}

\subsection{Axis 2: temporal structure}

A contract can reference one fixed event, multiple asynchronously resolving legs, a recurring sequence of events, or a persistent process without a natural event horizon. A rolling contract belongs on this axis: it changes the active constituent through time but does not define a new payoff geometry by itself.

\subsection{Axis 3: settlement structure}

Settlement can be terminal cash settlement, physical claim conversion, staged leg crystallization, rolling replacement, or periodic flow without principal settlement. A funding-only structure belongs here and is better described as an event-linked flow swap unless it also has a separately defined mark and terminal principal convention.

\subsection{Axis 4: venue and oracle composition}

The same geometry can use one venue and one oracle, several markets on one venue, or cross-venue and cross-oracle inputs. Composition determines source synchronization, dispute dependencies, fallback order, and whether a single contract has one finality clock or a vector of clocks.

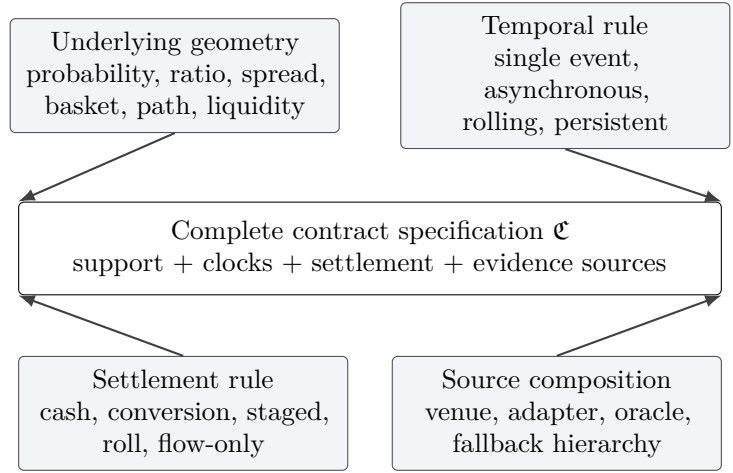
\begin{figure}[ht]
\centering
\begin{tikzpicture}[every node/.style={font=\small}]
  \node[seriesbox, text width=4.0cm] (geom) {Underlying geometry\\probability, ratio, spread,\\basket, path, liquidity};
  \node[seriesbox, right=8mm of geom, text width=4.0cm] (time) {Temporal rule\\single event, asynchronous,\\rolling, persistent};
  \node[draw=black, rounded corners=2pt, fill=white, below=8mm of $(geom.south)!0.5!(time.south)$, anchor=north, text width=8.8cm, align=center, inner sep=7pt] (contract) {Complete contract specification $\mathfrak C$\\support + clocks + settlement + evidence sources};
  \node[seriesbox, below=8mm of contract.south west, anchor=north west, text width=4.0cm] (settle) {Settlement rule\\cash, conversion, staged,\\roll, flow-only};
  \node[seriesbox, below=8mm of contract.south east, anchor=north east, text width=4.0cm] (comp) {Source composition\\venue, adapter, oracle,\\fallback hierarchy};
  \draw[seriesarrow] (geom.south) -- (contract.north west);
  \draw[seriesarrow] (time.south) -- (contract.north east);
  \draw[seriesarrow] (settle.north) -- (contract.south west);
  \draw[seriesarrow] (comp.north) -- (contract.south east);
\end{tikzpicture}
\caption{The four axes are orthogonal. A rolling spread, a cross-venue basket, or a funding-only liquidity contract is a composition of choices, not an additional primitive category.}
\label{fig:four-axes}
\end{figure}

\begin{criterion}[Composition completeness]
\label{crit:composition}
A composite design is evaluable only if every selected axis supplies a machine-readable rule and every axis-specific necessary condition is satisfied. A control that addresses one axis cannot repair an undefined or failed rule on another axis.
\end{criterion}

This criterion is the organizing principle for the remainder of the paper.

\section{Conditional-Probability Contracts}
\label{sec:conditional}

Let $u_t$ be the price of a joint claim $A\cap B$ and $v_t$ the price of the conditioning claim $B$. When both are treated as coherent probability-like quantities and $v_t>0$, the natural ratio is
\[
  c_t=\frac{u_t}{v_t},\qquad 0\le u_t\le v_t\le1.
\]
The ratio is not a complete derivative specification because $B=0$ leaves the conditional event undefined.

\begin{definition}[Conditional-ratio contract]
A conditional-ratio contract has geometry $g(u,v)=u/v$ on the domain $v>0$ and a separate settlement convention for states in which the conditioning event resolves false or $v$ reaches a defined floor.
\end{definition}

\subsection{Denominator sensitivity}

\begin{proposition}[Small-denominator sensitivity]
\label{prop:conditional-sensitivity}
For $c(u,v)=u/v$ on $0<u<v$, the gradient is
\[
  \nabla c(u,v)=\left(\frac{1}{v},-\frac{u}{v^2}\right).
\]
For any fixed $c_0\in(0,1)$ and sequence $(u_n,v_n)=(c_0v_n,v_n)$ with $v_n\downarrow0$, $\|\nabla c(u_n,v_n)\|\to\infty$. Thus bounded absolute measurement or hedge errors in either leg can induce unbounded local error in the ratio as the conditioning price approaches zero.
\end{proposition}

\begin{proof}
Substitution gives
\[
  \|\nabla c(c_0v,v)\|=\frac{\sqrt{1+c_0^2}}{v},
\]
which diverges as $v\downarrow0$.
\end{proof}

The proposition does not say that the ratio itself leaves $[0,1]$ along coherent states. It says that the map from two observed or hedgeable legs to the ratio becomes ill-conditioned. Coherent co-movement along the ray $u=c_0v$ can leave the ratio unchanged; independent quote, basis, latency, or execution errors are amplified.

\subsection{Settlement when the condition fails}

There is no unique no-arbitrage terminal value for $P(A\mid B)$ when $B=0$. The contract must choose a convention. Common choices are:

\begin{table}[ht]
\centering
\caption{Possible $B=0$ conventions. None follows from conditional probability alone.}
\label{tab:conditional-conventions}
\small
\begin{tabularx}{\textwidth}{L{3.1cm}YY}
\toprule
Convention & Benefit & Principal problem \\
\midrule
Void and return a defined account value & Avoids inventing a conditional outcome & Requires a precise variation-margin unwind rule \\
Fixed fallback value & Deterministic and easy to verify & Creates an arbitrary jump and distributional transfer \\
Pre-event time-weighted average & Reduces last-tick dependence & Remains manipulable and depends on window liquidity \\
Last valid ratio & Minimal additional machinery & Maximally exposed to final-observation manipulation \\
Joint-claim settlement $Y_A Y_B$ & Always defined & Changes the instrument from conditional probability to a joint event \\
\bottomrule
\end{tabularx}
\end{table}

A denominator floor $v_t\ge\varepsilon$ is a trading rule, not a settlement rule. It can stop new exposure before the ratio becomes ill-conditioned, but it does not determine what existing positions receive if $B$ resolves false.

\subsection{Resolution order and oracle composition}

If $A$ resolves before $B$, the semantic object $P(A\mid B)$ remains contingent on $B$. If $B$ resolves true first, the contract reduces to a single-event claim on $A$. If the legs use different oracles or rule sources, the contract requires a finality vector rather than one resolution timestamp. These choices alter both margin and manipulation surfaces.

\subsection{Empirical evaluability}

The ratio is directly observable only when the joint and conditioning claims are both listed or when a logical relation identifies the numerator exactly. Constructing $P(A\mid B)$ from two marginals alone requires a dependence model and is model-conditional evidence, not direct measurement. For mutually exclusive event groups, ratios such as $P(A_i\mid \neg A_j)=P(A_i)/(1-P(A_j))$ are observable only if the group semantics and identifiers are exact.

\section{Event Spreads and Asynchronous Finality}
\label{sec:spread}

Let two event-market references be $p_t^{(A)},p_t^{(B)}\in[0,1]$. The spread geometry is
\[
  S_t=p_t^{(A)}-p_t^{(B)}\in[-1,1].
\]
If both legs settle to $Y_A,Y_B\in\{0,1\}$, terminal support is $\{-1,0,1\}$, subject to any logical restrictions on the joint outcome set.

\subsection{First and second finality}

Let $\tau_A$ and $\tau_B$ be the leg-finality times used by the contract. Suppose $\tau_A<\tau_B$.

\begin{proposition}[Residual single-leg equivalence]
\label{prop:spread-residual}
After leg $A$ becomes final and before leg $B$ becomes final,
\[
  S_t=Y_A-p_t^{(B)}.
\]
For any $t_1<t_2$ in this interval,
\[
  S_{t_2}-S_{t_1}=-(p_{t_2}^{(B)}-p_{t_1}^{(B)}).
\]
Hence a long spread position has the same incremental price exposure as a short position in leg $B$, up to the constant $Y_A$.
\end{proposition}

\begin{proof}
Substitute the fixed value $p_t^{(A)}=Y_A$ after $\tau_A$ and take differences.
\end{proof}

The proposition gives a state transition for margin and funding. A two-leg engine before first finality should become a one-leg engine after first finality. Keeping the original two-leg parameters is not conservative by construction; it is simply the wrong exposure map.

\subsection{Settlement designs}

Three designs are coherent but not equivalent:

\begin{enumerate}
  \item \textbf{Wait for both legs.} The derivative remains open through the finality gap and settles at $Y_A-Y_B$.
  \item \textbf{Crystallize the first leg.} The first leg becomes a fixed cash component while the residual leg remains margined.
  \item \textbf{Convert into constituent claims.} The derivative is transformed into offsetting funded claims or another deliverable portfolio, subject to exact conversion identities.
\end{enumerate}

An early close at $\tau_A$ using the contemporaneous $p_{\tau_A}^{(B)}$ is a fourth design, but it transfers final-leg basis and liquidity risk into a closing-price rule.

\begin{figure}[ht]
\centering
\begin{tikzpicture}[node distance=18mm, every node/.style={font=\small}]
  \node[seriesbox, minimum width=3.1cm] (both) {Both legs active\\$S_t=p_t^{(A)}-p_t^{(B)}$};
  \node[seriesbox, right=of both, minimum width=3.5cm] (one) {First leg final\\$S_t=Y_A-p_t^{(B)}$\\single-leg residual};
  \node[seriesbox, right=of one, minimum width=3.1cm] (final) {Both legs final\\$S_\tau=Y_A-Y_B$};
  \draw[seriesarrow] (both) -- node[midway,above=2mm,fill=white,inner sep=1pt]{first finality} (one);
  \draw[seriesarrow] (one) -- node[midway,above=2mm,fill=white,inner sep=1pt]{second finality} (final);
  \draw[decorate,decoration={brace,amplitude=5pt,mirror},thick] ($(one.south west)+(0,-5mm)$) -- ($(final.south east)+(0,-5mm)$) node[midway,below=7mm]{finality gap};
\end{tikzpicture}
\caption{Asynchronous spread finality. The contract changes risk class after the first leg settles; it is not merely ``half resolved.''}
\label{fig:spread-finality}
\end{figure}
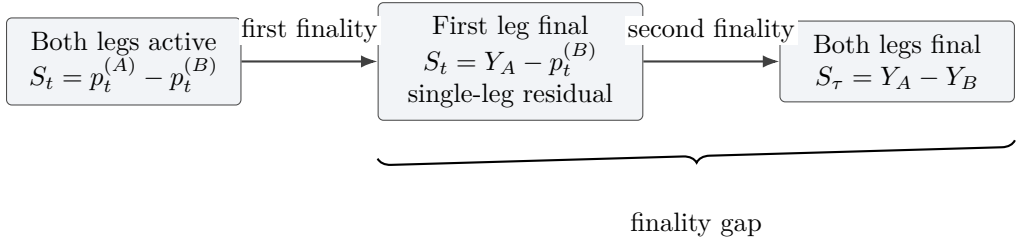

\subsection{Support and margin}

The spread support alone does not determine collateral. By \cref{prop:support-shortfall}, a long entered at $S_0$ has worst-case loss $x(S_0+1)$ when $-1$ is attainable; a short has worst-case loss $x(1-S_0)$ when $1$ is attainable. If logical constraints exclude an extreme joint outcome, the exact feasible support should replace $[-1,1]$.

A simultaneous terminal move can be as large as two support units from one endpoint to the other. A single-leg resolution jump is bounded by one. These are different scenarios and should not share one jump parameter.

\subsection{Empirical evaluability}

Same-venue spreads require synchronized leg prices, exact semantic pairing, per-leg finality, and sufficient overlapping depth. Cross-venue spreads additionally require clock normalization, fee and collateral conversion, and independent venue-status records. A fixed-path replay can evaluate a chosen rule mechanically, but it does not identify how a new spread market would change constituent liquidity.

\section{Event Baskets}
\label{sec:basket}

Let $k$ event legs have prices $p_t^{(i)}\in[0,1]$, weights $w_i\ge0$ with $\sum_iw_i=1$, and feasible terminal outcome set $\mathcal Y\subseteq\{0,1\}^k$. The normalized basket is
\[
  B_t=\sum_{i=1}^{k}w_i p_t^{(i)},
\]
with terminal values $w\cdot y$ for $y\in\mathcal Y$.

\begin{proposition}[Basket support and worst-case loss]
\label{prop:basket-support}
Define
\[
  b_{\min}=\min_{y\in\mathcal Y}w\cdot y,
  \qquad
  b_{\max}=\max_{y\in\mathcal Y}w\cdot y.
\]
Then the terminal basket support is contained in $[b_{\min},b_{\max}]$. A long entered at $B_0$ has worst-case loss $x(B_0-b_{\min})$ and a short has worst-case loss $x(b_{\max}-B_0)$. Increasing the number of legs does not by itself reduce either bound.
\end{proposition}

\begin{proof}
The first statement follows by evaluating the linear map $w\cdot y$ on the finite feasible set. The loss formulas follow from \cref{prop:support-shortfall}. The number of coordinates does not determine the extrema: if all-zero and all-one outcomes remain feasible, the bounds remain $0$ and $1$ for every $k$.
\end{proof}

This corrects a common diversification shortcut. A basket can have a smoother empirical distribution under weak dependence and diffuse weights, but that is a probabilistic assumption. It is not implied by the definition.

\subsection{Staggered resolutions}

When leg $i$ resolves, its contribution becomes the constant $w_iY_i$. The active basket is therefore a sequence of affine lower-dimensional baskets. Margin should be recomputed on the remaining feasible joint support, not mechanically reduced by the resolved weight. Logical dependence can make the residual support larger or smaller than a weight-only rule suggests.

\subsection{Constituent and weight governance}

A deployable basket must freeze or version:

\begin{itemize}
  \item constituent eligibility and replacement;
  \item weights and any rebalancing schedule;
  \item treatment of voided, delayed, or disputed legs;
  \item source hierarchy for each leg;
  \item concentration and common-factor limits.
\end{itemize}

Drop-on-resolution and weight renormalization change the contract geometry and can create jumps unrelated to event information. Keeping resolved legs at fixed payout preserves the original linear claim and is the cleanest baseline.

\subsection{Funding and execution}

Boundary funding applies to a normalized basket only when the basket can approach zero or one and the design seeks relative-basis control there. Cross-leg hedge execution is limited by shared depth and asynchronous fills. Summing standalone per-leg depth overstates executable basket capacity when the same liquidity or capital supports several legs.

\subsection{Empirical evaluability}

Evaluation requires a pre-specified constituent registry, exact weights, synchronized observations, feasible-outcome constraints, per-leg finality, and a shared-book execution model. Without the feasible joint outcome set, the support result and margin calibration are incomplete.

\section{Path-Functional Contracts: Variation and Entropy}
\label{sec:path-functionals}

The original taxonomy grouped volatility and entropy because both are functions of a probability process. Their terminal behavior is not the same and they should be treated separately.

\subsection{Realized variation}

Fix a sampling grid $t_0<\cdots<t_m$ and define realized quadratic variation
\[
  V=\sum_{j=1}^{m}(p_{t_j}-p_{t_{j-1}})^2.
\]
A rolling version replaces the full grid by a moving window. The contract is not specified until it fixes the grid, missing-observation rule, stale-price rule, inclusion or exclusion of the terminal jump, and annualization convention. With a fixed finite grid and $p\in[0,1]$, $0\le V\le m$, but the bound is a sampling artifact rather than an economic terminal value.

A realized-variation contract avoids directional binary settlement: the same event outcome can produce different $V$. It does not avoid event-time jump risk if the terminal price change is included. Excluding the terminal jump removes the largest event move by convention and changes the economic object.

\subsection{Entropy}

For a binary price-like input $p\in[0,1]$, define Shannon entropy in natural units by
\[
  H(p)=-p\log p-(1-p)\log(1-p),
\]
with $0\log0=0$.

\begin{proposition}[Deterministic entropy collapse]
\label{prop:entropy-collapse}
If the terminal event price is set to $Y\in\{0,1\}$, then $H(Y)=0$ for either outcome. The maximum possible pre-terminal entropy is $\log2$ at $p=1/2$. Therefore an entropy-level contract settled at event finality has a deterministic terminal target of zero and can experience a terminal drop as large as $\log2$.
\end{proposition}

\begin{proof}
Direct evaluation gives $H(0)=H(1)=0$. Differentiation gives $H'(p)=\log((1-p)/p)$, so the unique interior maximum is $p=1/2$ with value $\log2$.
\end{proof}

Entropy removes outcome direction, not terminal maturity. A long entropy position held to finality loses its entire entry level absent offsetting cash flows. It therefore does not escape the support-dependent account test.

\subsection{Hedgeability and observability}

Neither realized variation nor entropy is a directly deliverable asset. A market maker must hedge with dynamic positions in the underlying event claim, options on the claim if they exist, or offsetting flow. Near finality, those hedges can be unavailable precisely when the path functional changes most.

Empirical evaluation is possible from high-frequency prices only after fixing a canonical observation process. Quote midpoint, executable price, last trade, and robust index can produce materially different variation and entropy series. The choice must be part of the contract, not an analyst preference selected after results are observed.

\section{Liquidity-Index Contracts and Reflexivity}
\label{sec:liquidity}

A liquidity-index contract references a microstructure statistic such as effective spread, executable depth, price impact, cancellation-adjusted depth, or a composite of these quantities. The apparent appeal is direct hedging of market-liquidity risk. The principal difficulty is endogeneity.

\begin{definition}[Liquidity-index contract]
Let $M_t$ be a defined venue-state object and $\ell_t=h(M_t)$ a deterministic, versioned statistic. A liquidity-index contract is a derivative whose cash flows depend on $\ell_t$ or its path.
\end{definition}

The definition requires a complete observation rule: side, depth band, order eligibility, stale-quote treatment, venue outages, hidden or off-chain orders, and aggregation across markets. A displayed-book statistic is not automatically executable liquidity.

\subsection{Fixed-path replay limitation}

\begin{proposition}[Non-identification under reflexivity]
\label{prop:reflexive-replay}
Let a proposed contract or policy $c$ change trader or liquidity-supplier behavior, so venue state is $M(c)$. A fixed-path replay that applies $c$ to the observed no-contract path $M(0)$ identifies the rule output $F(c;M(0))$. It does not identify the deployment output $F(c;M(c))$ unless either $M(c)=M(0)$ over the relevant state variables or a valid model identifies the response map $c\mapsto M(c)$.
\end{proposition}

\begin{proof}
The two estimands differ by their second argument. Equality requires invariance of the relevant venue state or an identified transformation connecting the states. A deterministic replay alone supplies neither.
\end{proof}

This limitation is especially severe when the contract pays on the same spread or depth that its hedging and liquidation trades alter.

\subsection{Manipulation and source availability}

A liquidity index can be moved by posting, cancelling, splitting, or withholding orders. Robust construction therefore requires order eligibility, persistence, self-trade controls, and manipulation-cost analysis. Paper~4 shows that address-level quote-lifecycle attribution is structurally unavailable from public on-chain fills on a Polymarket-class hybrid CLOB, while fill attribution remains available \citep{paper4_microstructure_2026}. A contract requiring public address-level cancellation history would therefore fail an observability gate on that venue.

\subsection{Settlement}

Liquidity has no natural event payoff. The contract must use a dated observation, rolling average, option-like trigger, or flow-only settlement. Calling it an event-linked perpetual does not solve the absence of a natural terminal rule.

\section{Rolling Structures and Flow-Only Contracts}
\label{sec:rolling-flow}

Rolling and flow-only structures are modifiers rather than new underlying geometries.

\subsection{Rolling event sequence}

Let $c=1,2,\ldots$ index recurring event contracts with active intervals $[a_c,b_c]$ and underlying $I_t^{(c)}$. A rolling rule selects an active constituent and defines a roll time $r_c\le b_c$ together with a transition from $c$ to $c+1$.

A complete roll specification requires:

\begin{itemize}
  \item eligibility of the successor event;
  \item overlap and observation windows;
  \item closing and opening price rules;
  \item treatment of missing successors;
  \item roll cash flow and accumulated funding;
  \item whether positions are automatically transferred or closed.
\end{itemize}

A rolling interface can remain available indefinitely, but each constituent still has finite finality. The roll does not eliminate terminal risk; it decides whether that risk is realized, converted, or transferred before each constituent ends.

\subsection{Roll discontinuity}

If $I_{r_c}^{(c)}\neq I_{r_c}^{(c+1)}$, a direct index substitution creates a roll jump unrelated to new event information. A back-adjusted series removes the visual jump but introduces a cumulative accounting adjustment. Futures-style cash settlement of the roll is transparent but realizes basis. The contract must choose one.

\subsection{Funding-only event swap}

A flow-only contract has periodic cash flow
\[
  \Delta C_j=x\,\varphi(Z_{t_j})\,\Delta t_j,
\]
where $Z_t$ is a versioned event-related statistic and $\varphi$ maps it to a payable rate. There is no principal exchange at event finality.

This object is closer to a swap on a flow than to a perpetual future on a deliverable index. It avoids terminal principal shortfall by definition, but creates counterparty-transfer and participation risk. The rate must be capped or otherwise collateralized, and the contract needs a close-out rule for accumulated unpaid flows.

Basis, disagreement, cross-venue divergence, and liquidity can each define $Z_t$, but they are different contracts. A generic label ``funding-only'' is insufficient.

\subsection{Empirical evaluability}

Rolling evaluation requires several complete event cycles and successor mappings. Flow-only evaluation requires the exact statistic, interval schedule, payer collateral, and close-out rule. Fixed-path replay is weakest when $Z_t$ is endogenous to participation in the flow contract.

\section{Corrected Inheritance and Cross-Cutting Controls}
\label{sec:crosscutting}

The revised inheritance map is based on properties, not product names.

\begin{table}[ht]
\centering
\caption{Risk inheritance by canonical design. ``Yes'' means the Paper~1 property applies under the stated support; ``modified'' means a new definition is required; ``no'' means the property is not structurally present.}
\label{tab:inheritance}
\scriptsize
\begin{tabularx}{\textwidth}{L{2.2cm}*{5}{>{\centering\arraybackslash}X}}
\toprule
Design & Terminal support test & Natural $0/1$ boundary funding & Asynchronous finality & Reflexive reference & Natural terminal rule \\
\midrule
Single probability & Yes & Yes & No & Usually no & Binary payout \\
Conditional ratio & Yes, convention-dependent & Yes & Yes & No & No when condition false \\
Event spread & Yes, support-dependent & Modified per leg & Yes & No & Two-leg terminal value \\
Normalized basket & Yes, joint-support-dependent & Possible at basket boundaries & Yes & No & Weighted terminal value \\
Realized variation & Yes, grid-dependent & No & Window-dependent & Possible & Dated/window rule required \\
Entropy level & Yes; terminal target zero & No relative-probability rule & Single event & No & Zero at finality \\
Liquidity index & Rule-dependent & No & Source-dependent & Yes & None \\
Rolling structure & Inherited per constituent & Inherited per constituent & Often & Depends on geometry & Roll rule \\
Flow-only swap & Transfer-solvency test & Target-dependent & Not necessarily & Depends on target & No principal settlement \\
\bottomrule
\end{tabularx}
\end{table}

\subsection{Control assignment}

Five design controls recur, but their inputs differ.

\begin{enumerate}
  \item \textbf{Support-aware collateral.} Compute margin from attainable terminal support and current position direction. Do not infer a loss bound from the number of legs.
  \item \textbf{State-aware exposure transition.} After a leg finalizes, recompute the derivative's affine exposure and margin state.
  \item \textbf{Economically bounded transfers.} Funding or flow payments must respect payer collateral and participation limits; clipping changes the convergence objective and should be explicit.
  \item \textbf{Definition-complete settlement.} Void, delay, missing source, and false conditioning states require contractual payouts or unwind rules.
  \item \textbf{Evidence-compatible evaluation.} A design whose reference is endogenous or unobservable cannot be validated by replaying an incomplete public data surface.
\end{enumerate}

\subsection{Finality vector}

For a multi-leg or multi-source contract, define
\[
  \mathbf t^{\mathrm{fin}}
  =\bigl(t_1^{\mathrm{oracle}},\ldots,t_k^{\mathrm{oracle}},
  t_1^{\mathrm{protocol}},\ldots,t_k^{\mathrm{protocol}}\bigr).
\]
A scalar ``resolution time'' is valid only when the settlement rule is a deterministic function of a single component or all relevant components are synchronized. Otherwise the state machine must preserve which components are final and which remain contestable.

\begin{criterion}[Most restrictive inherited condition]
For a composed design, admission and margin must satisfy every necessary condition generated by its axes. If a rolling cross-venue spread inherits a roll gap, asynchronous finality, and cross-venue source failure, the engine cannot select only the most convenient one.
\end{criterion}

\section{Empirical Evaluability}
\label{sec:evaluability}

Taxonomy is useful only if it distinguishes what can be measured from what can merely be defined.

\subsection{Evidence ladder}

\begin{definition}[Evaluability grades]
A proposed variant is assigned the highest grade it satisfies:
\begin{description}[leftmargin=1.4cm,style=nextline]
  \item[E0] Formal payoff and state-machine specification only.
  \item[E1] Reference series reconstructable from exact observed inputs.
  \item[E2] Terminal and intermediate finality states exactly joinable.
  \item[E3] Mechanical replay feasible on fixed observed paths with explicit invariance limits.
  \item[E4] Behavioral or equilibrium response identified from live, randomized, structural, or otherwise credible evidence.
\end{description}
\end{definition}

A large dataset does not raise a variant above E1 if the required semantic joins or finality states are missing.

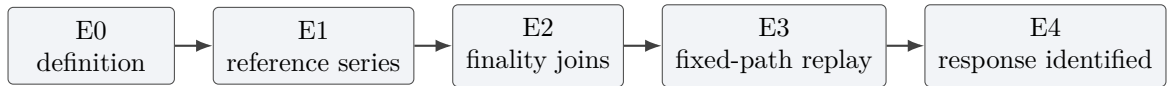
\begin{figure}[ht]
\centering
\begin{tikzpicture}[node distance=5mm, every node/.style={font=\small}]
  \node[seriesbox, minimum width=2.2cm] (e0) {E0\\definition};
  \node[seriesbox, right=of e0, minimum width=2.2cm] (e1) {E1\\reference series};
  \node[seriesbox, right=of e1, minimum width=2.2cm] (e2) {E2\\finality joins};
  \node[seriesbox, right=of e2, minimum width=2.2cm] (e3) {E3\\fixed-path replay};
  \node[seriesbox, right=of e3, minimum width=2.2cm] (e4) {E4\\response identified};
  \draw[seriesarrow] (e0) -- (e1);
  \draw[seriesarrow] (e1) -- (e2);
  \draw[seriesarrow] (e2) -- (e3);
  \draw[seriesarrow] (e3) -- (e4);
\end{tikzpicture}
\caption{Evaluability is cumulative. Mechanical replay is not an equilibrium test, and an observable reference series is not a complete finality dataset.}
\label{fig:evidence-ladder}
\end{figure}

\subsection{Minimum data by design}

\begin{table}[ht]
\centering
\caption{Minimum evidence needed for canonical designs. The table identifies necessary inputs, not a claim that they are currently available.}
\label{tab:data-requirements}
\small
\begin{tabularx}{\textwidth}{L{2.4cm}YY}
\toprule
Design & Minimum reference data & Additional blocking evidence \\
\midrule
Conditional ratio & Joint and conditioning legs with exact IDs & False-condition settlement and per-leg finality \\
Spread & Synchronized leg prices and depth & Semantic pairing and first/second finality \\
Basket & Constituent panel, weights, joint constraints & Shared execution and per-leg source states \\
Variation / entropy & Canonical high-frequency index series & Grid/window rule and terminal-print treatment \\
Liquidity index & Full statistic inputs at required granularity & Quote lifecycle, manipulation controls, endogeneity \\
Rolling & Repeated event cycles and successor registry & Roll execution and missing-successor states \\
Flow-only & Exact target statistic and payment schedule & Counterparty collateral and close-out data \\
\bottomrule
\end{tabularx}
\end{table}

\subsection{Current Polymarket-class boundaries}

The PMXT and on-chain data used in the companion papers support exact fills, market-level book observations, and many event identifiers, but public address-level quote placement and cancellation are unavailable on a hybrid CLOB \citep{paper4_microstructure_2026}. This blocks some liquidity-index and participant-response claims even when fill data are large. Cross-platform variants require an additional venue archive and clock reconciliation. Rolling variants require multiple complete cycles. Conditional contracts require exact joint-market semantics, not title similarity.

\subsection{Recommended evaluation order}

The most defensible sequence is:

\begin{enumerate}
  \item same-venue probability and spread contracts with exact finality;
  \item exact-joint conditional contracts and fixed-weight baskets;
  \item path-functional contracts with frozen observation rules;
  \item rolling versions after several complete cycles;
  \item liquidity-index and flow-only contracts only after reflexivity and participant-response designs are specified.
\end{enumerate}

The ordering prioritizes exact identity and finality before increasingly endogenous references.

\section{Design Implications and Research Boundaries}
\label{sec:implications}

\subsection{There is no universal event-perpetual engine}

The taxonomy eliminates the premise that one margin, funding, and halt template can be inherited across all event-linked contracts. The base account identity generalizes, but its inputs change. A conditional contract adds definition failure. A spread adds a finality gap. A basket adds joint-support and shared-liquidity constraints. Entropy has a deterministic zero endpoint. A liquidity index adds reflexivity. Rolling and flow-only structures change the temporal and settlement objects.

\subsection{Product naming should follow the balance sheet}

A contract with periodic cash flows and no principal settlement is an event-linked flow swap. A contract that finances real outcome tokens is a credit product, not merely a synthetic perpetual. A rolling sequence is a programme of finite contracts unless a legally and economically continuous obligation is defined. Clear naming is a risk control because it determines which collateral and finality questions are visible.

\subsection{Connection to manipulation and microstructure}

The taxonomy supplies the objects on which manipulation analysis should operate. Ratio denominators, cross-leg finality, basket weights, liquidity statistics, and roll prices create distinct attack surfaces. Paper~3 develops the incentive and regulatory analysis; Paper~4 identifies which participant and quote-lifecycle claims are empirically observable \citep{paper3_manipulation_2026,paper4_microstructure_2026}.

\subsection{No new empirical claims}

This revision does not evaluate the variants, simulate demand, or recalibrate Paper~1. Construction recipes and evidence grades are research specifications. A future empirical paper should register one contract definition, one source map, and one primary estimand rather than treating the taxonomy as a menu to be tuned after results.

\subsection{Limitations}

The taxonomy is not exhaustive. Options, tranche structures, portfolio margin, physically delivered conversions, and multi-outcome categorical claims can be represented by the axes but are not analyzed in detail. Legal enforceability, user demand, capital formation, and welfare are outside scope. The formal results assume that stated settlement values and source identities are enforceable. No result establishes production safety.

\section{Conclusion}
\label{sec:conclusion}

Event-linked derivatives should be classified by the mathematical and institutional rules that determine their payoffs, not by a shared interface label. The four-axis taxonomy separates underlying geometry, temporal structure, settlement, and source composition. This separation corrects several overbroad inheritance claims in the earlier version.

The formal results are intentionally narrow. Terminal shortfall depends on exact support and collateral. Conditional ratios become ill-conditioned near a vanishing denominator. An asynchronously resolving spread becomes a single-leg residual after first finality. Basket risk depends on feasible joint outcomes, not leg count. Entropy collapses to zero at event finality. Liquidity-index replay does not identify deployment behavior when the reference responds to the contract itself.

These results yield one operational rule: a composite design inherits the most restrictive necessary conditions of all its axes. No margin rule repairs an undefined settlement convention, no halt repairs missing collateral, and no large archive repairs an unobserved or endogenous reference process. The appropriate next step is therefore not to list every conceivable variant. It is to select one fully specified point in the taxonomy and collect evidence at the grade required by its claims.

\appendix
\section{Variant Specification Checklist}
\label{app:checklist}

Before empirical work begins, a proposed contract should answer the following questions.

\begin{enumerate}
  \item What are the exact reference variables and their support?
  \item Which observations are executable prices, which are marks, and which are documentary inputs?
  \item What are the active event, roll, proposal, oracle, protocol, and redemption clocks?
  \item What is the terminal or periodic cash-flow rule in every admissible state?
  \item What happens when a conditioning event is false, a source is missing, a leg is voided, or finality is delayed?
  \item Which venue, adapter, oracle, and rule source is authoritative for each input?
  \item What collateral is available to each position direction under the exact terminal support?
  \item Which controls act on collateral, execution, definition failure, finality, and reflexivity?
  \item Is the reference path plausibly invariant to deployment? If not, what identifies the response?
  \item What data and joins are required for the intended evidence grade?
\end{enumerate}

A specification that cannot answer one of these questions remains E0, regardless of implementation detail elsewhere.

\section{Additional Formal Details}
\label{app:formal}

\subsection{General affine residual after partial finality}

The spread result extends to any affine multi-leg contract
\[
  X_t=a_0+\sum_{i=1}^{k}a_i p_t^{(i)}.
\]
If a subset $F$ of legs has finalized, then
\[
  X_t=\underbrace{a_0+\sum_{i\in F}a_iY_i}_{\text{constant}}
  +\sum_{i\notin F}a_i p_t^{(i)}.
\]
The active delta vector is obtained by deleting finalized coordinates. This identity is useful for baskets and spreads and requires no stochastic assumptions.

\subsection{Conditional-ratio perturbation bound}

For perturbations $(\delta u,\delta v)$ with $v>|\delta v|>0$,
\[
  \frac{u+\delta u}{v+\delta v}-\frac{u}{v}
  =\frac{v\delta u-u\delta v}{v(v+\delta v)}.
\]
Therefore
\[
  \left|\Delta c\right|
  \le
  \frac{v|\delta u|+u|\delta v|}{v(v-|\delta v|)}.
\]
The bound shows explicitly why fixed absolute errors become material as $v$ shrinks.

\subsection{Basket support under logical constraints}

If outcomes satisfy linear or logical constraints, $\mathcal Y$ should be constructed before margin is set. For mutually exclusive and exhaustive outcomes with equal weights, the basket terminal value may be constant; in that case a derivative on the basket level has no terminal directional uncertainty even though each leg is binary. This is another reason leg count alone is uninformative.

\section{Notation Registry}
\label{app:notation}

\begin{table}[ht]
\centering
\caption{Canonical notation for this paper.}
\begin{tabularx}{\textwidth}{L{2.4cm}Y}
\toprule
Symbol & Meaning \\
\midrule
$p_t^{(i)}$ & Observable reference price or index for leg $i$ \\
$q_t$ & Derivative mark price \\
$Y_i$ & Terminal payout of leg $i$ \\
$Z$ & Contract terminal settlement value \\
$\mathcal Z$ & Attainable terminal settlement support \\
$x$ & Position size \\
$C_t^{\mathrm{net}}$ & Net collateral excluding the position's own price profit and loss \\
$\tau_i$ & Finality time used by the contract for leg $i$ \\
$g$ & Underlying-geometry map \\
$\mathcal T$ & Temporal rule \\
$\mathcal S$ & Settlement rule \\
$\mathcal C$ & Venue-oracle composition rule \\
$\mathfrak C$ & Complete event-linked derivative specification \\
\bottomrule
\end{tabularx}
\end{table}

The word \emph{oracle} is qualified throughout. A resolution oracle determines an event outcome; an optimization oracle is a computational subroutine. Only the first appears in the contract-composition axis.

\section{Revision and Reproducibility Record}
\label{app:reproducibility}

This major revision uses the published Paper~2 source as a documentary input and the corrected Paper~1 account mechanics as its formal base. No empirical dataset was queried and no numerical result was recomputed.

The revision changes claims rather than outputs:

\begin{itemize}
  \item the flat seven-variant list is replaced by four orthogonal axes;
  \item inheritance is tied to support and settlement assumptions;
  \item the conditional denominator statement is proved rather than left as a conjecture;
  \item spread finality is expressed as an exact residual-exposure identity;
  \item basket diversification claims are restricted to explicit joint-outcome assumptions;
  \item realized variation and entropy are separated;
  \item rolling and flow-only structures are reclassified as modifiers;
  \item fixed-path replay limits are made explicit for reflexive references.
\end{itemize}

The source package includes the revision notes, formal audit, shared series style, build report, and cryptographic manifest.

\printbibliography[heading=bibintoc,title={References}]
\end{document}